\documentclass{IEEEtran4PSCC}

\ifCLASSINFOpdf
   \usepackage[pdftex]{graphicx}
\else
   \usepackage[dvips]{graphicx}
\fi
\usepackage[table]{xcolor}

\usepackage[cmex10]{amsmath}
\usepackage{amssymb}
\usepackage{bm}
\usepackage{amsmath}
\usepackage{amsthm}
\usepackage{tabularx}
\usepackage{multirow} 
\usepackage{makecell} 
\usepackage{array}
\usepackage{url}
\newtheorem{theorem}{Theorem}[section]

\newtheorem{proposition}[theorem]{Proposition}

\usepackage{array}
\makeatletter
\let\old@ps@headings\ps@headings
\let\old@ps@IEEEtitlepagestyle\ps@IEEEtitlepagestyle
\def\psccfooter#1{%
    \def\ps@headings{%
        \old@ps@headings%
        \def\@oddfoot{\strut\hfill#1\hfill\strut}%
        \def\@evenfoot{\strut\hfill#1\hfill\strut}%
    }%
    \def\ps@IEEEtitlepagestyle{%
        \old@ps@IEEEtitlepagestyle%
        \def\@oddfoot{\strut\hfill#1\hfill\strut}%
        \def\@evenfoot{\strut\hfill#1\hfill\strut}%
    }%
    \ps@headings%
}
\makeatother

\psccfooter{%
        \parbox{\textwidth}{\hrulefill \\ \small{24th Power Systems Computation Conference} \hfill \begin{minipage}{0.2\textwidth}\centering \vspace*{4pt} \includegraphics[scale=0.06]{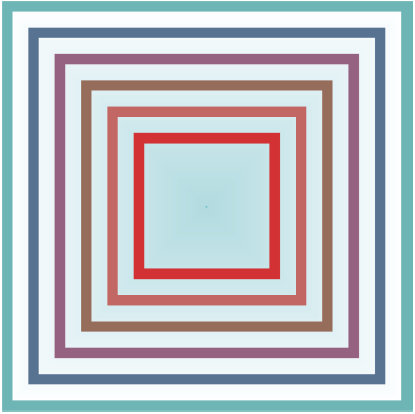}\\\small{PSCC 2026} \end{minipage} \hfill \small{Limassol, Cyprus --- June 8-12, 2026}}%
}

\begin{document}
%
\title{Dispatchable Current Source  Virtual Oscillator Control Achieving Global Stability}

\author{
\IEEEauthorblockN{Kehao Zhuang, Linbin Huang, Huanhai Xin}
\IEEEauthorblockA{College of Electrical Engineering\\Zhejiang University\\Hangzhou, China\\
\{zhuangkh, hlinbin, xinhh\}@zju.edu.cn}
\and
\IEEEauthorblockN{Xiuqiang He}
\IEEEauthorblockA{Department of Automation\\Tsinghua University\\
Beijing, China\\
\{hxq19\}@tsinghua.org.cn}
\and
\IEEEauthorblockN{Verena Häberle, Florian Dörfler}
\IEEEauthorblockA{Automatic Control Laboratory\\ETH Zürich\\
Zürich, Switzerland\\
\{verenhae, dorfler\}@ethz.ch}
}


\maketitle

\begin{abstract}
This work introduces a novel dispatchable current source virtual oscillator control (dCVOC) scheme for grid-following (GFL) converters, which exhibits duality with dispatchable virtual oscillator control (dVOC) in two ways: a) the current frequency is generated through reactive power control, similar to a PLL ; b) the current magnitude reference is generated through active power control. We formally prove that our proposed control always admits a steady-state equilibrium and ensures global stability under reasonable conditions on grid and converter parameters, even when considering LVRT and current saturation constraints. Our approach avoids low-voltage transients and weak grid instability, which is not the case for conventional GFL control. The effectiveness of our proposed control is verified through high-fidelity electromagnetic transient simulations.
\end{abstract}


\thanksto{\noindent Submitted to the 24th Power Systems Computation Conference (PSCC 2026).}

\section{Introduction}

With the increasing renewable energy sources integrated into the power grid through phase-locked loop (PLL)-based grid-following (GFL) converters~\cite{20}, synchronous generators (SGs) are gradually being displaced, resulting in weak grid characteristics with low short circuit ratio (SCR)~\cite{10:cigre}. Unfortunately, GFL converters are highly prone to instability under such weak grid conditions~\cite{Huqi:torque_method}. Numerous instability incidents involving converters have been reported worldwide, often leading to large-scale tripping events~\cite{NERC}. Under small disturbances, GFL converters typically exhibit instability manifested as sub- and super-synchronous oscillations~\cite{kehao:dual_axis}. Under large disturbances, GFL converters often lose synchronism or experience voltage collapse due to the absence of the equilibrium point~\cite{multiPLLtransient}. Therefore, addressing these challenges by enhancing GFL control to ensure stable power delivery remains a critical research focus. 

A common mitigation strategy for suppressing oscillation is to employ auxiliary control loops and lead-lag compensators to selectively shape the converter’s phase margin~\cite{xiongfei:syn_overview,Linbin:PLL_sy}. However, since these oscillations span a wide frequency range and their dominant frequencies vary with grid conditions, phase compensators designed for specific frequency bands may fail to provide universally effective mitigation~\cite{xiongfei:syn_overview}. 



In recent years, various transient enhancement control strategies have also been proposed for GFL converters. These include reducing active current proportionally to the percentage of the frequency variations of PLL, adjusting active and reactive currents to match line impedance angle for improved stability, and freezing the PLL integrator to increase damping during transients, etc \cite{xiongfei:syn_overview}. Although these approaches have demonstrated effectiveness in numerous case studies, none can theoretically guarantee the global stability of GFL converters. Moreover, during low voltage ride through (LVRT) events, grid codes impose specific requirements on reactive current support~\cite{lvrt:IEC}. Therefore, transient enhancement control must ensure both stability and compliance with grid code requirements. 

In contrast to GFL converters, grid-forming (GFM) converters exhibit low sensitivity to grid strength and maintain excellent small-signal (local) stability even in grids with very low SCR \cite{xiongfei:syn_overview}. However, conventional droop-controlled converters that directly emulate synchronous machines still suffer from transient stability issues and often perform poorly under large disturbances \cite{Kehao:saturation}. To address this, dispatchable virtual oscillator control (dVOC) has been proposed \cite{M.colombino:dvoc}. dVOC offers remarkable advantages in synchronization and voltage stability by introducing a tunable power coupling angle and a nonlinear voltage term \cite{hybrid_angle:tayyebi}, which together ensure almost global stability \cite{Xiuqiang:crossforming}. Its dynamic response not only surpasses that of conventional droop control but also removes the restrictive assumptions of purely inductive networks, fixed resistance to inductance ratios \cite{grob:dvoc}, or constant voltages \cite{xiuqiangdvoc:conletter}. The effectiveness of dVOC has been extensively validated through numerous experimental studies \cite{Hyu:dvoc,dvoc_xiuqiang:tcns}. 

\begin{table*}
\centering
\caption{The comparison of GFM and GFL}
\footnotesize
\begin{tabularx}{\linewidth}{
    >{\centering\arraybackslash}m{3cm}
    >{\centering\arraybackslash}m{4.8cm}
    >{\centering\arraybackslash}m{9cm}}
\hline\hline
 & \multicolumn{1}{c}{GFM droop control} & \multicolumn{1}{c}{GFL control} \\ 
\arrayrulecolor{gray!50}  \cline{1-3}  
\multicolumn{1}{c}{Controller} & \multicolumn{1}{c}{$
    
\begin{aligned}
\dot{\theta}_u&=\omega_0+k_p\left(p_{\rm ref}-p\right)\,,\\
\dot{u}&=k_p\left(q_{\rm ref}-q\right)+k_v\left(u_{\rm ref}-u\right)\,.
\end{aligned}$} & \multicolumn{1}{c}{$\begin{aligned}
\dot{\theta}_i&=\omega_0+k_{\rm pllp}u_q+k_{\rm plli}\int{u_qdt}\,,\\
{i}&=k_{\rm pp}\left(p_{\rm ref}-p\right)+k_{\rm pi}\int{\left(p_{\rm ref}-p\right)dt}\,.
\end{aligned}$
}\\
\cline{1-3}  
\multirow{2}{*}{Parameters} & $k_p \in \mathbb{R}_{>0}$: power droop gain, $k_v \in \mathbb{R}_{>0}$: voltage droop gain. & $k_{\rm pllp},k_{\rm plli} \in \mathbb{R}_{>0}$: proportional and integral coefficient of the PLL, $k_{\rm pp},k_{\rm pi} \in \mathbb{R}_{>0}$: proportional and integral coefficient of active power control.\\
\cline{2-3}
& \multicolumn{2}{>{\centering\arraybackslash}p{13.8cm}}{$\omega_0$: nominal frequency; \; $\theta_u,u$: voltage angle and magnitude at the converter terminal; \; $\theta_i,i$: output current angle and magnitude of the converter; \;$p,q$: active power and reactive power of the converter; \;subscript $_{\rm ref}$: set reference value; \;subscript $_{d},_{q}$: $d$- or $q$-axis component  in converter's local $dq$ frame component.}\\
 \cline{1-3}  
Source type & Controllable voltage source that actively provides voltage and frequency support.  &  Controllable current source delivering stable active power while passively following the grid.\\ 
\cline{1-3}  
{Frequency response} & Active power-voltage frequency droop  &   Reactive power-current frequency PI control at unity power factor ($q=u_qi$) \\
 \cline{1-3}  
{Magnitude adjustment} & Voltage magnitude response to reactive power variation  & Current magnitude response to active power variation\\
 \cline{1-3}  
{Stability} & Stable across wide grid strength, even in low-SCR weak grids.  &  Prone to instability in low-SCR weak grids. \\
\arrayrulecolor{black}\hline\hline
\end{tabularx}
\label{Table-gflgfm}
\end{table*}

\begin{figure*}
\centerline{\includegraphics[scale=0.7]{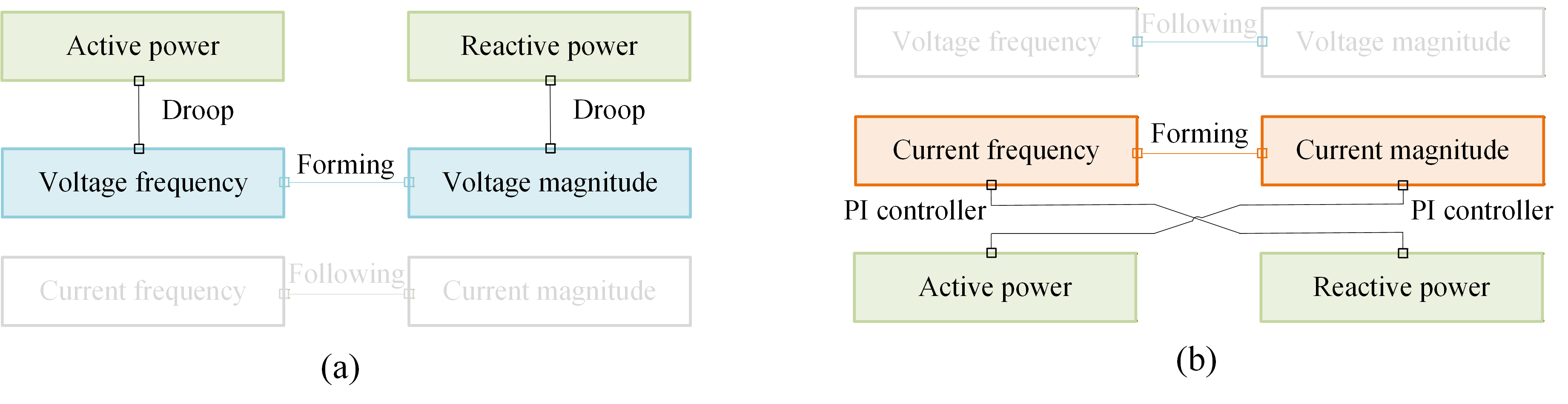}}
        \caption{The duality between GFM and GFL (a) GFM (voltage-forming and current-following), (b) GFL (voltage-following and current-forming).}
        \label{dual}
\end{figure*}

However, replacing all GFL converters with GFM would incur prohibitively high costs. Consequently, GFL converters are expected to remain a long-term component of future power systems. This motivates the development of a robust GFL control strategy. To address the weak-grid instability in GFL converters, we observe that, by leveraging the duality between GFM and GFL modes~\cite{duality:gu}, the dVOC~\cite{3:XIUQIANGDVOC}, which exhibits near-global stability, can be seamlessly extended to a GFL framework. This approach is compatible with conventional GFL control architectures and does not require complex control structures. We term this control strategy {\em dispatchable current source virtual oscillator control (dCVOC)}, with the specific contributions summarized as follows.

Under certain parameter settings, dCVOC achieves stable synchronization via a PLL and regulates active power through current magnitude control, ensuring stable power delivery under steady-state conditions. Moreover, we demonstrate that, regardless of grid conditions, dCVOC always possesses an equilibrium point and can achieve global asymptotic stability, even when accounting for LVRT grid code requirements and current saturation. These parameter conditions are easily satisfied and do not require detailed knowledge of the grid. Furthermore, the effectiveness of the proposed dCVOC is validated through EMT simulations using a detailed switching model in MATLAB/Simulink.

\section{The Control Architecture of dCVOC} \label{sec1}
\subsection{The Duality of GFL and GFM }

We take classical droop control and PLL-based control as representative examples of GFM and GFL, respectively, to illustrate their duality. The dynamic equations of the GFM and GFL controls are given in Table~\ref{Table-gflgfm}, respectively. 

{\bf Duality between GFM and GFL modes:} Traditionally, studies on GFM and GFL converters have primarily focused on the voltage and current source perspectives. More recently, a complementary viewpoint based on duality has been introduced, emphasizing that the GFM mode can be regarded as {\em voltage-forming} and {\em current-following}, whereas the GFL mode can be regarded as {\em current-forming} and {\em voltage-following}~\cite{duality:gu,Xiuqiang:crossforming}. Specifically, GFM independently regulates the voltage magnitude and frequency (angle), while its current depends on the external grid. In contrast, GFL independently regulates the current magnitude and frequency (angle), while its voltage depends on the external grid. Given the inherent duality between voltage and current, GFM and GFL naturally exhibit both structural and functional duality, as shown in Fig.~\ref{dual}.

{\bf Dual frequency response:} A GFM converter establishes an active power–frequency droop relationship. When a load disturbance induces a grid frequency deviation, this droop mechanism enables the converter to provide active power support, thereby contributing to frequency stabilization. In contrast, a GFL converter determines its current frequency from $u_q$ via a high-bandwidth PLL, effectively tracking the grid frequency. Typically operating at unity power factor ($i_q=0,q=u_qi_d-u_di_q=u_qi_d\propto u_q$), the GFL converter can therefore be equivalently interpreted as regulating reactive power to track the grid frequency. 

{\bf Dual magnitude adjustment:} A GFM converter regulates its voltage magnitude through variations in reactive power and incorporates a voltage regulation term $k_v\left(u_{\rm ref}-u\right)$ to support system voltage. In contrast, a GFL converter adjusts its current magnitude via changes in active power, ensuring stable active power transfer.

GFM converters are insensitive to grid strength and can maintain stable operation under weak grid conditions. Building upon droop control, the dVOC scheme has been proposed to further enhance stability, ensuring near-global convergence of GFM converters during transient events. In contrast, conventional GFL converters are prone to instability in weak grid. By leveraging the duality between GFM and GFL modes, we aim to develop an advanced GFL control strategy with stability characteristics (near-global stability) similar to those of dVOC. 

\subsection{From dVOC to dCVOC}
We first revisit the dVOC~\cite{3:XIUQIANGDVOC} and then derive its dual formulation by exploiting the voltage–current duality. The dVOC dynamics are  given by
\begin{equation} \label{eq3:dVOC}
\begin{aligned}
\dot{\theta}_u&=\omega_0+k_p\left(p^{\rm ref}_\varphi-p_\varphi\right)\,,\\
\frac{\dot{u}}{u}&=k_p\left(q^{\rm ref}_\varphi-q_\varphi\right)+k_v\frac{u^2_{\rm ref}-u^2}{u^2_{\rm ref}}\,,
\end{aligned}
\end{equation}
where, $p^{{\rm ref}}_\varphi+{\rm j}q^{{\rm ref}}_\varphi=e^{{\rm j}(\pi/2-\varphi)}(p_{\rm ref}+{\rm j}q_{\rm ref})/u^2_{\rm ref}$, and $p_\varphi+{\rm j}q_\varphi=e^{{\rm j}(\pi/2-\varphi)}(p+{\rm j}q)/u^2$. The angle $\varphi$ is used to rotate the power to match the  the network impedance characteristics. In an inductive-dominated network, set $\varphi=\pi/2$, consistent with classical droop control. In addition, the nonlinear term $1/u^2$ plays a crucial role in ensuring stability. 

Following the duality framework presented in Section~\ref{sec1}, the mapping between the control inputs and the magnitude–frequency outputs is interchanged. The proposed dCVOC is shown in Fig.~\ref{control}, and its dynamics are given by
\begin{equation} \label{eq4:dCVOC}
\begin{aligned}
\dot{\theta}_i&=\omega_0+k_p\left(q_\varphi -q^{\rm ref}\right)+k_{\rm plli}\int{i \left(q_\varphi-q^{\rm ref}_\varphi\right)dt}\,,\\
\frac{\dot{i}}{i}&=k_p\left(p_\varphi-p^{\rm ref}_\varphi\right)\,,
\end{aligned}
\end{equation}
where $p^{\rm ref}_\varphi+{\rm j}{q^{\rm ref}_\varphi}=e^{j(\pi/2-\varphi)}(p_{\rm ref}+{\rm j}q_{\rm ref})/i^2_{\rm ref}$, $p_\varphi+j{q_\varphi}=e^{j(\pi/2-\varphi)}(p+jq)/i^2$, and $\varphi$ carries the same meaning as in Eq.~\eqref{eq3:dVOC}.Dual to dVOC, dCVOC also includes a nonlinear term, namely the power divided by the square of the current $i^2$. 
\begin{figure}
\centerline{\includegraphics[width=0.9\linewidth]{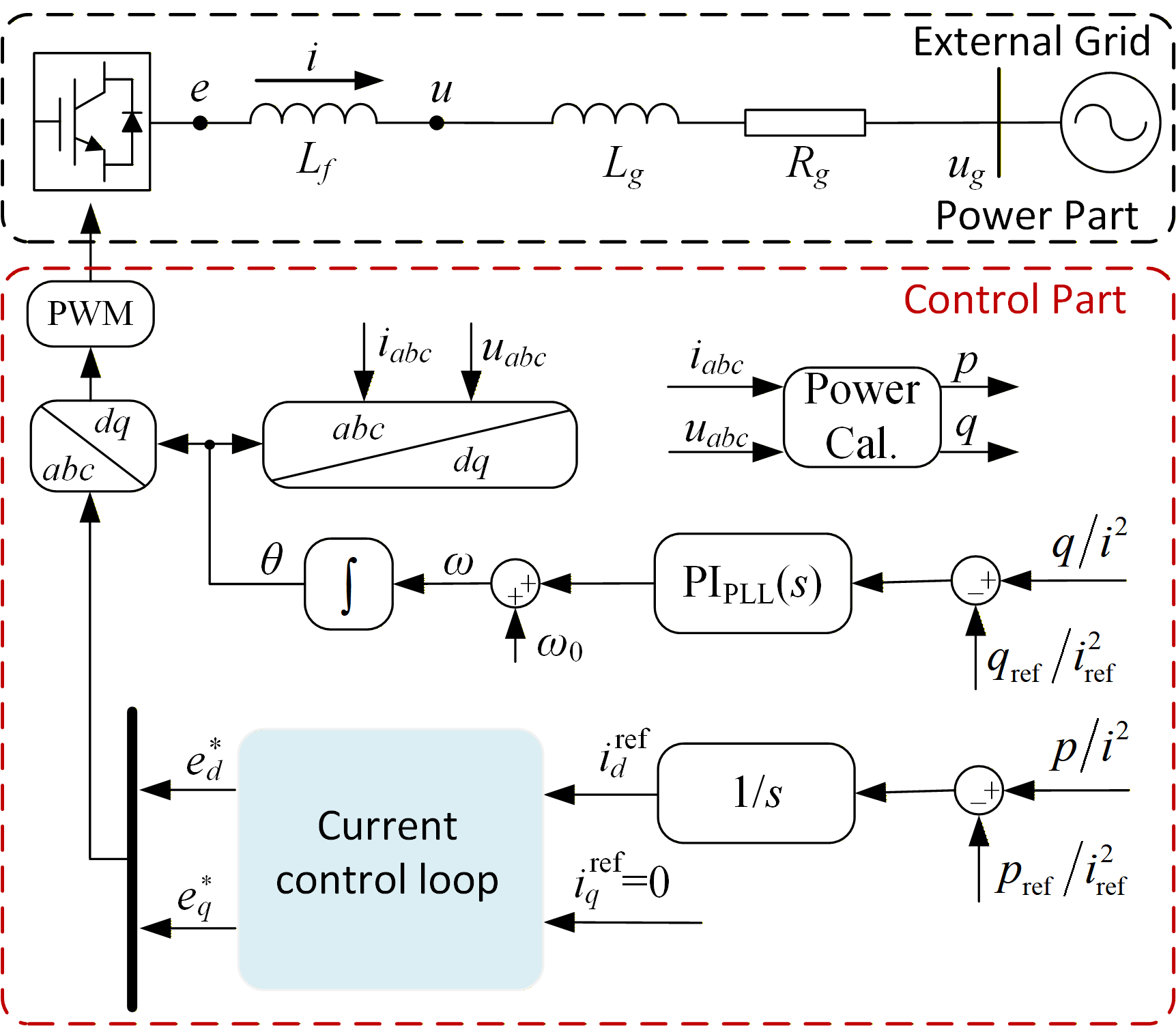}}
        \caption{The control architecture of dCVOC}
        \label{control}
\end{figure}

The proposed dCVOC is formulated by exploiting the duality illustrated in Fig.~\ref{dual}, where $q_\varphi$ (i.e., reactive power divided by current magnitude for $\varphi=\pi/2$) governs the current frequency and $p_\varphi$ (i.e., active power divided by current magnitude for $\varphi=\pi/2$) governs current magnitude, thereby extending the dVOC concept to a GFL context. 

The proposed dCVOC is also compatible with conventional GFL control architectures. Setting $\varphi=\pi/2$ and $q_{\rm ref}=0$ yields $q_\varphi=u_q/i$. In this configuration, both dCVOC and conventional GFL control rely on a PLL (PI controller) for synchronization, with $u_q=0$ at steady state. The difference is that dCVOC incorporates a variable gain of $k_p/i$. In both cases, active power is regulated through current magnitude, but dCVOC provides enhanced stability and adaptability due to its variable-gain structure, which is proved in Section~\ref{sec3}.

\section{The Stability Analysis of dCVOC} \label{sec3}
\subsection{The Existence of Equilibrium Points and Global Stability}
We consider a single converter connected to the grid modeled as a stiff voltage source, as shown in Fig.~\ref{control}. It will be convenient to analyze \eqref{eq4:dCVOC} in different coordinates aligned with those chosen for the dVOC analysis. Taking the grid voltage angle $\theta_g$ as the reference (i.e., $\theta_g=0{\rm rad}$ and $\omega_g=\dot{\theta_g}$) and letting $\vec{i}=i_{\alpha}+{\rm j}i_{\beta}=ie^{j\delta_i}$ ($\alpha\beta$ is a global frame) with ${\delta}_i={\theta_i}-{\theta_g}$, then we have
\begin{equation} \label{eq5:vector}
\begin{aligned}
\dot{\vec{i}}=e^{j\delta_i}\dot{i}+{\rm j}\left(ie^{j\delta_i}\dot{\delta_i}\right)=\left(\dfrac{\dot{i}}{i}+{\rm j}\dot{\delta_i}\right)\vec{i}\,.
\end{aligned}
\end{equation}

Due to $\dot{\delta}_i=\dot{\theta}_i-\omega_g$, substituting \eqref{eq4:dCVOC} into \eqref{eq5:vector} and eliminating $\dot{i}/i,\dot{\delta}$ yields
\begin{equation} \label{eq5:dvoc_vector}
\begin{aligned}
\dot{\vec{i}}=k_p\left[\left(p_{\varphi}-p^{\rm ref}_{\varphi}\right)+{\rm j}\left(q_{\varphi}-q^{\rm ref}_{\varphi}+\omega_\Delta \right)\right]\vec{i}\,,
\end{aligned}
\end{equation}
where, $\omega_\Delta=\frac{1}{k_p}\left(k_{\rm plli}\int{i\left(q_\varphi-q^{\rm ref}_\varphi\right)dt}+\omega_0-\omega_g\right)$, $\omega_0$ and $\omega_g$ are both constants.

The output power is 
\begin{equation} \label{eq5:power}
\begin{gathered}
p+{\rm j}q=(u_\alpha+{\rm j}u_\beta)\vec{i}^*\,,\\
(p_\varphi+{\rm j}q_{\varphi})\vec{i}=e^{{\rm j}\left(\frac{\pi}{2}-\varphi\right)}\dfrac{p+{\rm j}q}{\vec{i}^*}=e^{{\rm j}\left(\frac{\pi}{2}-\varphi\right)}(u_{\alpha}+{\rm j}u_\beta)\,,
\end{gathered}
\end{equation}
where $^*$ denotes the complex conjugate, $u_\alpha+{\rm j}u_\beta=ue^{j\delta_u}$ with $\delta_u=\theta_u-\theta_g$.

For a complex number $a+{\rm j}b$, it can be represented in vector $\begin{bmatrix} a & b\end{bmatrix}^\top$ or matrix $\left[\begin{smallmatrix} a & -b \\ b & a\end{smallmatrix}\right]$ (Conversely, $\vec{()}$ represents the complex form of the above matrix or vector). Therefore, substituting \eqref{eq5:power} into \eqref{eq5:dvoc_vector} yields
\begin{equation} \label{eq3:dCVOC_xy}
\left\{
\begin{aligned}
\begin{bmatrix} \dot{i_\alpha} \\ \dot{i_\beta} \end{bmatrix}
&= k_p \left({\bm S}_\varphi-{\bm S}^{\rm ref}_\Delta\right)
\begin{bmatrix} i_\alpha \\ i_\beta \end{bmatrix} \\[0.8em]
&= k_p \left( e^{J(\pi/2-\varphi)}
\begin{bmatrix} u_\alpha \\ u_\beta \end{bmatrix}
- {\bm S}^{\rm ref}_\Delta
\begin{bmatrix} i_\alpha \\ i_\beta \end{bmatrix} \right) , \\[0.8em]
\dot{\omega}_\Delta
&= \dfrac{k_{\rm plli}}{k_p}i(q_\varphi - q^{\rm ref}_\varphi)\,,
\end{aligned}
\right.
\end{equation}
where, ${\bm S}_\varphi:=\left[\begin{smallmatrix} p_\varphi & -q_\varphi  \\ q_\varphi&  p_\varphi \end{smallmatrix}\right]\,,{\bm S}^{\rm ref}_{\Delta}:=\left[\begin{smallmatrix} p^{\rm ref}_\varphi & -\left(q^{\rm ref}_\varphi-\omega_\Delta \right)  \\ q^{\rm ref}_\varphi-\omega_\Delta & p^{\rm ref}_\varphi \end{smallmatrix}\right]$, $e^{J\theta}=\left[\begin{smallmatrix} \cos\theta & -\sin\theta\\ \sin\theta & \cos\theta \end{smallmatrix}\right]$ .

The network electromechanical equations are given as
\begin{equation} \label{eq3:line}
\begin{aligned}
\begin{bmatrix} {u_\alpha} \\ {u_\beta} \end{bmatrix}-
\begin{bmatrix} {u_g} \\ 0 \end{bmatrix}=Z_ge^{J\varphi_g}
\begin{bmatrix} {i_\alpha} \\ {i_\beta} \end{bmatrix}
\,,
\end{aligned}
\end{equation}
where, $u_g$ is the grid voltage magnitude, $Z_g=\sqrt{R^2_g+L^2_g}$ is the line impedance magnitude, and $\varphi_g$ is the impedance angle. If $\varphi_g$ is known, one typically sets $\varphi = \varphi_g$ to compensate for the impedance angle. 

Combining \eqref{eq5:power} to~\eqref{eq3:line}, and eliminating $q_{\varphi}, \left[u_\alpha\;\;u_\beta\right]^\top$, the interconnected converter and network dynamics are expressed as
\begin{equation} \label{eq3:close}
\begin{cases}
\begin{bmatrix} \dot{i_\alpha} \\ \dot{i_\beta} \end{bmatrix}=
k_p\left(\left(Z_ge^{J(\frac{\pi}{2}+\varphi_g-\varphi)}-{\bm S}^{\rm ref}_\Delta\right)\begin{bmatrix} {i_\alpha} \\ {i_\beta} \end{bmatrix}+u_{g,\varphi}\right)
 \\
 \dot{\omega}_\Delta=\dfrac{k_{\rm plli}}{k_p}i\left(\Im{\left(\dfrac{\vec{u}_{g,\varphi}\vec{i}^*}{i^2}\right)}+Z_g\sin(\frac{\pi}{2}+\varphi_g-\varphi)- q^{\rm ref}_\varphi\right)
\,,
\end{cases}
\end{equation}
where $i_{\alpha\beta}=\begin{bmatrix}  i_\alpha & i_\beta  \end{bmatrix}^\top$, $u_{g,\varphi}=e^{J(\frac{\pi}{2}-\varphi)}\begin{bmatrix} {u_g} & 0 \end{bmatrix}^\top$, and $\Im$ denotes imaginary part.

By setting the differential equations in ~\eqref{eq3:dCVOC_xy} and~\eqref{eq3:close} to zero, the unique equilibrium point can be obtained,
\begin{equation} \label{eq3:equlibrium}
\begin{cases}
\begin{bmatrix} {i_{\alpha,s}} \\ {i_{\beta,s}} \end{bmatrix}=-\left(e^{J\varphi_g}Z_g-{\bm S}_{\rm ref}\right)^{-1}\begin{bmatrix} {u_g} \\ 0 \end{bmatrix}\,,\\
\omega_{\Delta,s}=0
\,,
\end{cases}
\end{equation}
where the subscript $_{s}$ denote the steady value at the equilibrium point, ${\bm S}_{\rm ref}:=\frac{1}{i^2_{\rm ref}}\left[\begin{smallmatrix} p_{\rm ref} & -q_{\rm ref} \\ q_{\rm ref} & p_{\rm ref}\end{smallmatrix}\right]$. And~\eqref{eq3:equlibrium} holds provided that the matrix $e^{J\varphi_g}Z_g-{\bm S}_{\rm ref}$ is invertible, a condition that is readily satisfied, ensuring the existence of a steady-state equilibrium.

At the equilibrium point, substituting \eqref{eq3:line} and \eqref{eq3:equlibrium} into \eqref{eq5:power} yields the steady-state output power and current magnitude,
\begin{equation} \label{eq3:equiPQ}
\begin{aligned}
p_s+jq_s&=u^2_s\frac{i^2_{\rm ref}(p_{\rm ref}+{\rm j}q_{\rm ref})}{p^2_{\rm ref}+q^2_{\rm ref}}\,,\\
i_s&=u_s\dfrac{i^2_{\rm ref}}{\sqrt{p^2_{\rm ref}+q^2_{\rm ref}}}
\,.
\end{aligned}
\end{equation}

If $i_{\rm ref}=\sqrt{p^2_{\rm ref}+q^2_{\rm ref}}$, \eqref{eq3:equiPQ} shows that under steady-state or small-disturbance conditions with $u_s\approx 1$pu, dCVOC delivers power according to the prescribed reference values. During large disturbances such as short-circuit faults, the terminal voltage drops, and dCVOC adaptively decreases its steady-state current magnitude $i_s=u_si_{\rm ref}$ in proportion to the voltage depression, thereby maintaining the existence of a feasible power flow solution.


Eq.~\eqref{eq3:equlibrium} provides the unique equilibrium point, and the parameter conditions ensuring its existence and global asymptotic stability can be further established. 

\begin{proposition}[{Global stability}]\label{globalstability}
Assume that \eqref{eq3:conditionstability} holds and there is $\epsilon_s$ so that $0<\epsilon<\epsilon_s$, where $\epsilon=k_{\rm plli}/k^2_p$. Then the unique equilibrium $(\vec{i}_{\alpha\beta,s},\omega_{\Delta,s})$ always exists and is globally asymptotically stable.
\begin{equation} \label{eq3:conditionstability}
\begin{aligned}
\sqrt{p^2_{\rm ref}+q^2_{\rm ref}}\cos{\left(\frac{\pi}{2}-\varphi\right)}>i^2_{\rm ref}Z_g\cos\left(\frac{\pi}{2}+\varphi_g-\varphi\right)
\,,  
\end{aligned}
\end{equation}
\end{proposition}

The parameter conditions in~\eqref{eq3:conditionstability} is readily satisfied, thereby enabling dCVOC to achieve global asymptotic stability. When no power rotation is applied, i.e. $\varphi=\pi/2$, \eqref{eq3:conditionstability} reduces to $p_{\rm ref}>i^2_{\rm ref}R_g$, meaning the active power must exceed the line power losses. When considering a power rotation angle, this is equivalent to requiring that the “virtual active power”  $\sqrt{p^2_{\rm ref}+q^2_{\rm ref}}\cos{\left(\frac{\pi}{2}-\varphi\right)}$ exceeds the “virtual line power losses” $i^2_{\rm ref}Z_g\cos\left(\frac{\pi}{2}+\varphi_g-\varphi\right)$. Moreover, if the power rotation angle fully compensates for the network impedance angle, i.e. $\varphi=\varphi_g$, setting $\sqrt{p^2_{\rm ref}+q^2_{\rm ref}}>0$ ensures that~\eqref{eq3:conditionstability} is always satisfied. 



\subsection{The Stability Analysis Considering Current Limitation and LVRT}
\textit{1) Current limitation.} Due to the converter’s overcurrent limitation, its output current cannot exceed the predefined maximum value $i_{\max}$. Otherwise, current saturation will occur, i.e. $i=i_{\rm max}$ and $\dot{i}=0$.

However, unlike dVOC, dCVOC directly controls current as a state variable, thereby ensuring stable operation even under current-limiting conditions. 

\begin{proposition}[{Stability under current limitation}]\label{currentsaturation}
If the parameter condition \eqref{eq3:conditionstability} for global asymptotic stability are satisfied, the system remains globally asymptotically stable even when current saturation is considered.
\end{proposition}
 \begin{proof}
Under current saturation, i.e., $i=i_{\max}$, $\dot{i}=0$, substituting it into the proof process in Appendix~\ref{globalstability}, the condition $\dot{V}<0$ still holds, thus ensuring global asymptotic stability.

\end{proof}

\textit{2) LVRT. }According to the grid code, during LVRT events, GFL converters are typically required to inject reactive power in proportion to the voltage sag. dCVOC can directly comply with the grid code by following the prescribed $q_{\rm ref}$
\begin{equation} \label{eq3:LVRTdCVOC}
\begin{aligned}
p_{\rm ref}&=\max\left(k_sp_{\min},\sqrt{(k_si_{\rm max})^2-q^2_{\rm ref}}\right)\,,{\rm if}\,u\leq 0.9 {\rm pu}\,,\\
q_{\rm ref}&=k_s\min \left(k_l\left(0.9-u\right),\sqrt{i^2_{\rm max}-p^2_{\min}}\right)\,,{\rm if}\,u\leq 0.9 {\rm pu}\,,
\end{aligned}
\end{equation}
where $p_{\min}$ is the minimum active power reference that guarantees the condition in Eq.~\eqref{eq3:conditionstability}. If the network impedance angle $\varphi_g$ can be accurately estimated such that $\varphi=\varphi_g$, then $p_{\min}$ can be set to zero. 

From the Eq.~\eqref{eq3:equiPQ}, it can be seen that dCVOC reduces the current according to the voltage drop. To fully exploit the reactive current support capability, the gain $k_s$ can be approximately chosen as a constant inversely proportional to the voltage, i.e., $k_s \approx 1/u$, to ensure that the reactive power during LVRT complies with grid code requirements. Therefore, by specifying the dCVOC references according to Eq.~\eqref{eq3:LVRTdCVOC} during LVRT, both stability and the grid code requirements can be simultaneously ensured. 

In summary, dCVOC ensures the existence of an equilibrium and guarantees global stability without requiring detailed grid information, while simultaneously satisfying LVRT and current-limiting constraints. 

\section{Simulation results}
The stability of the dCVOC is verified through time-domain simulations in MATLAB/Simulink based on the switching-circuit model in Fig.~\ref{control}. We set $k_p=20$ and $k_{\rm plli}=20$ at per-unit, resulting in a small constant $\epsilon=0.05$. The current reference is set as $i_{\rm ref}=\sqrt{p^2_{\rm ref}+q^2_{\rm ref}}$. In steady state, set $p_{\rm ref}=1$pu and $q_{\rm ref}=0$. The grid nominal frequency is $\omega_0=50{\rm Hz}$ and two grid impedance cases are considered—a large magnitude $Z_g$ and a low $R_g/L_g$ ratio (corresponding to a high-voltage level transmission line), as well as a small magnitude $Z_g$ and a high  $R_g/L_g$ ratio  (corresponding to a low-voltage level transmission line)—to demonstrate the stability of dCVOC under different grid conditions.

In all simulation cases, the grid voltage $u_g$ is set to drop at $t=1s$ and recover at $t = 2s$. In Case 1, the line impedance is set to $L_g=0.65$pu (${\rm SCR}\approx1.5$), $R_g=0.05$pu, where Case 1.1 corresponds to a grid voltage drop to 0.8pu and Case 1.2 to 0.2 pu. In Case 2, the line impedance is set to $L_g=0.25$pu, $R_g=0.2$pu, where Case 2.1 corresponds to a grid voltage drop to 0.8pu and Case 2.2 to 0.2 pu.

For comparison, we first show the behavior of a conventional GFL coverter as in Table~\ref{Table-gflgfm}. As shown in Fig.~\ref{time} (a), (b), the conventional GFL converter exhibits pronounced instability under weak grid conditions. During severe voltage sags, it becomes unstable within the LVRT period and fails to recover normal operation after the grid voltage is restored. Under mild voltage sags, the converter still shows weakly damped oscillations that eventually converge during LVRT, and instability still occurs during voltage recovery. In Fig.~\ref{time} (c), (d), under strong grid conditions with high $R_g/L_g$ ratio, the conventional GFL converter demonstrates significantly better stability: even though it may lose stability during LVRT due to the absence of an equilibrium point, it can still regain stable operation once the grid voltage recovers.

Compared with the conventional GFL converter, our proposed dCVOC exhibits markedly superior performance. As shown in Fig.~\ref{time}, dCVOC ensures the existence of an equilibrium point, maintains stable operation during LVRT and smoothly returns to the steady-state equilibrium after grid recovery—regardless of grid strength or fault severity. Furthermore, during LVRT, dCVOC also provides reactive power support proportional to the voltage sag according to the grid code, thereby contributing to grid voltage recovery.

\begin{figure}[h]
 \vspace{-2mm}
\centerline{\includegraphics[width=0.9\linewidth]{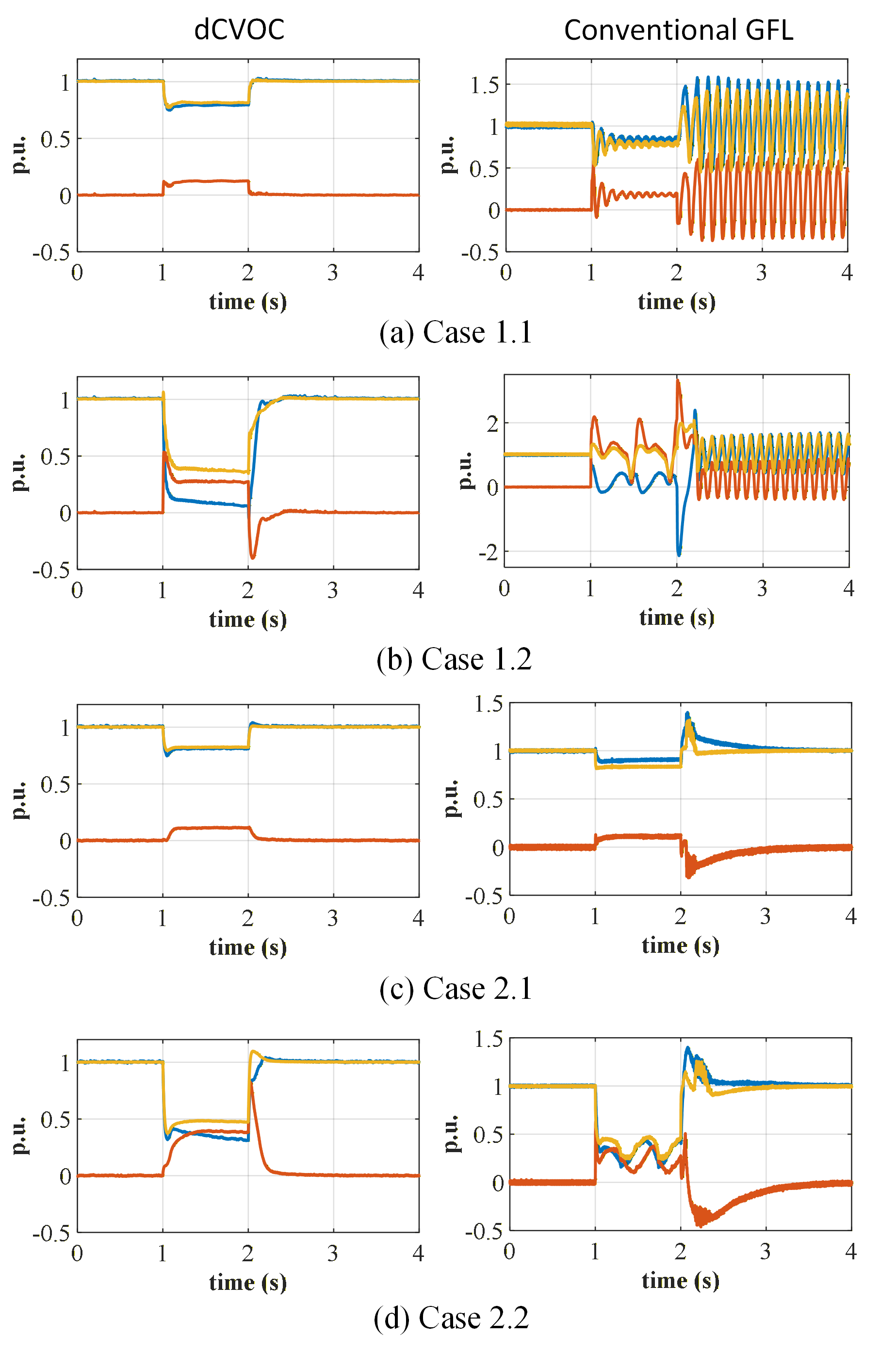}}
        \caption{Time domain responses of our proposed dCVOC and a conventional GFL.   \textcolor[HTML]{0070C0}{\textbf{---}} active power $p$;
        \textcolor[HTML]{C00000}{\textbf{---}} reactive power $q$; \textcolor[HTML]{DAA520 }{\textbf{---}}  voltage magnitude $u$;.}
        \label{time}
\end{figure}


\section{Conclusion}
This paper leverages the duality between GFM and GFL converters to extend dVOC to a GFL context, termed dCVOC. The main conclusions are as follows:

\begin{enumerate}
    \item dCVOC is fully compatible with the conventional GFL control framework and can ensure phase-locking and stable unit power factor operation while transmitting active power under steady-state conditions.
    \item Under easily satisfied explicit parameter conditions, dCVOC guarantees the existence of an equilibrium point under arbitrary grid conditions. Physically, this is because dCVOC can adjust its output power according to the voltage, ensuring that a power flow solution always exists.
    \item Moreover, based on singular perturbation theory, it is demonstrated that the dCVOC preserves global asymptotic stability even when accounting for LVRT control strategies and current saturation, provided that the proposed explicit parameter conditions are satisfied. 
\end{enumerate}

The proposed dCVOC addresses the instability issues of conventional GFL converters in weak grids, offering a promising approach for integrating large-scale power electronic devices into the grid. Future work will focus on enhancing dCVOC control and providing stability guarantees for multi-converter systems, taking DC bus voltage dynamics into account.

 {\appendices
\section{Proof of Proposition~\ref{globalstability}}
\begin{proof}
   The equilibrium point of \eqref{eq3:equlibrium} exists uniquely if and only if $e^{J\varphi_g}Z_g-{\bm S}_{\rm ref}$ is nonsingular, i.e.
\begin{equation} \label{eq3:conditionofeq}
i^2_{\rm ref}(R_g+{\rm j}L_g)\neq p_{\rm ref}+{\rm j}q_{\rm ref}
\,.
\end{equation}

For convenience in analyzing stability, we shift the coordinate system so that the equilibrium point is at the origin. Let $\tau=\frac{k_{\rm plli}}{k_p}t$, $\epsilon =k_{\rm plli}/k^2_p$, $x=\omega_\Delta-\omega_{\Delta,s}$, $z={i}_{\alpha\beta}-{i}_{\alpha\beta,s}$, $b=Z_g\sin(\frac{\pi}{2}+\varphi_g-\varphi)-q^{\rm ref}_\varphi$, $\|\|$ denotes $\ell_2$ norm, Eq.~\eqref{eq3:close} can be rewritten as
\begin{equation} \label{eq3:timescale}
\begin{cases}
 \dfrac{dx}{d\tau}=\|z+{i}_{\alpha\beta,s}\|\left({\dfrac{\Im{(\vec{u}_{g,\varphi}\vec{z}^*+\vec{u}_{g,\varphi}\vec{i}^*_{\alpha\beta,s})}}{\|z+{i}_{\alpha\beta,s}\|^2}}+b\right)\\
\epsilon \dfrac{dz}{d\tau}=\left(e^{J(\frac{\pi}{2}+\varphi_g-\varphi)}Z_g-{\bm S}^{\rm ref}_\Delta(x) \right)(z+{i}_{\alpha\beta,s})+u_{g,\varphi}\,,
\end{cases}
\end{equation}

Under the premise that an equilibrium point exists, according to singular perturbation theory~\cite{khalil_nonlinear_2002}, when $0<\epsilon<\epsilon_s$ with $\epsilon_s$ being a sufficiently small constant, the system can be decomposed into a boundary layer system and a reduced-order slow system. 

    In the boundary layer system $\epsilon\frac{dz}{d\tau}$, if $x$ is treated as a constant, the equilibrium point of boundary-layer system is given by
    \begin{equation} \label{eq3:equiofboundary}
    \begin{aligned}
        z_s&=\left({\bm S}^{\rm ref}_\Delta(x)-Z_ge^{J(\frac{\pi}{2}+\varphi_g-\varphi)} \right)^{-1}u_{g,\varphi}-i_{\alpha\beta,s}\\
        &=:z_0u_{g,\varphi}-i_{\alpha\beta,s}
        \,.
    \end{aligned}
\end{equation}
with
$$
z_0u_{g,\varphi}=e^{J(\frac{\pi}{2}-\varphi)}\frac{u_g}{a^2+(b+x)^2}\begin{bmatrix}-a\\b+x\end{bmatrix}\,,
$$
where  $a=Z_g\cos{(\frac{\pi}{2}+\varphi_g-\varphi)}-p^{\rm ref}_\varphi$, and $b=Z_g\sin(\frac{\pi}{2}+\varphi_g-\varphi)-q^{\rm ref}_\varphi$.

Let $y=z-z_s$, $\frac{dx}{d\tau}$ can be rewritten as
\begin{equation} \label{eq3:slow}
\begin{aligned}
 \dfrac{dx}{d\tau}=\|y+z_0u_{g,\varphi}\|\left({\dfrac{\Im{(\vec{u}_{g,\varphi}\vec{y}^*)+u^2_g\Im(\vec{z}^*_0)}}{\|y+z_0u_{g,\varphi}\|^2}}+b\right)=:f(x,y)
\end{aligned}
\end{equation}

If we consider the reduced slow system $\frac{dx}{d\tau}$, $z=z_s$ is treated as the quasi-steady state solution of the boundary layer subsystem. The reduced-order slow system can thus be expressed as
\begin{equation} \label{eq3:reducedsystem}
\begin{aligned}
    \dfrac{dx}{d\tau}&=\|z_0u_{g,\varphi}\|\left(\dfrac{u^2_g\Im(\vec{z}^*_0)}{\|z_0u_{g,\varphi}\|^2}+b\right)=-\|z_0u_{g,\varphi}\|x\\
&=-\dfrac{u_gx}{a^2+(b+x^2)}=:f_s(x)\,.
\end{aligned}
\end{equation}

The boundary-layer system can be rewritten as
\begin{equation} \label{eq3:boundarysystem}
\epsilon \dfrac{dy}{d\tau}=\left(e^{J(\pi/2+\varphi_g-\varphi)}Z_g-{\bm S}^{\rm ref}_\Delta(x)\right)y-\epsilon\dfrac{\partial z_s}{\partial x}f(x,y)\,.
\end{equation}
where
\begin{equation} \label{eq3:ztox}
\begin{aligned}
    \dfrac{\partial z_s}{\partial x}=\dfrac{\partial z_0u_{g,\varphi}}{\partial x}=-\dfrac{u_ge^{J(\frac{\pi}{2}-\varphi)}}{\left(a^2+(b+x)^2\right)^2}\begin{bmatrix} 2a(b+x) \\ a^2-(b+x)^2\end{bmatrix}\,,
\end{aligned}
\end{equation}


The positive-definite and radially unbounded Lyapunov candidate functions, denoted as $V=V_1+V_2$, is constructed for systems of \eqref{eq3:slow} and \eqref{eq3:boundarysystem},
\begin{equation} \label{eq3:lyapunov}
V_1=\frac{k_p}{2k_{\rm plli}}x^2 ,\,\,V_2=\frac{k_p\epsilon}{2k_{\rm plli}}y^\top y.
\end{equation}


The time derivative of $V_1$ along the reduced-order system of \eqref{eq3:reducedsystem} dynamics is
\begin{equation} \label{eq3:dotlyapunov_slow}
\begin{aligned}
\dfrac{\partial V_1}{\partial x}f_s(x)=xf_s(x)=-\|z_0u_{g,\varphi}\|x^2\leq-\alpha_1 \psi_1^2(x)\,,
\end{aligned}
\end{equation}
where $\psi_1(x)=|x|$, and $0<\alpha_1 \leq \|z_0u_{g,\varphi}\|$ is a sufficiently small constant.

In \eqref{eq3:ztox} and \eqref{eq3:slow}, the following bounds hold: $\|\frac{\partial z_s}{\partial x}\|\leq u_g$, $f(x,y)\leq u_g+\|y+z_0u_{g,\varphi}\|b$. The time derivative of $V_2$ along the boundary-layer dynamics is
\begin{equation} \label{eq3:dotlyapunov_V2}
\begin{aligned}
\dot{V}_2=-\left[ p^{\rm ref}_\varphi-\Re\left(Z_ge^{j(\pi/2+\varphi_g-\varphi)}\right)\right]\|y\|^2-\epsilon y^\top\dfrac{\partial z_s}{\partial x}f(x,y)
\,.
\end{aligned}
\end{equation}

The first term $-\left[ p^{\rm ref}_\varphi-\Re\left(Z_ge^{j(\pi/2+\varphi_g-\varphi)}\right)\right]\|y\|^2\leq-\alpha_2\psi^2_2(y)$ in \eqref{eq3:dotlyapunov_V2} with $\psi_2(y)=\|y\|$. And it should satisfy $0 \leq \alpha_2\leq p^{\rm ref}_\varphi-\Re\left(Z_ge^{j(\pi/2+\varphi_g-\varphi)}\right)$, i.e. satisfy the condition of \eqref{eq3:conditionstability}. 

Because $\|\epsilon y^\top\frac{\partial z_s}{\partial x}f(x,y)\|\leq \epsilon\|y\|(u^2_g+\|y+z_0u_{g,\varphi}\|u_gb)$ remains bounded, there exists a sufficiently small constant $\epsilon$ and sufficiently larger constants $\beta_2,\gamma$ so that $\|\epsilon y^\top\dfrac{\partial z_s}{\partial x}f(x,y)\|\leq \beta_2\psi_1(x)\psi_2(y)+\gamma\psi^2_2(y)$. Furthermore, because $(f(x,y)-f_s(x))$ remains bounded, there exists a sufficiently larger constants $\beta_1$ so that $\frac{\partial V_1}{\partial x}(f(x,y)-f_s(x))\leq \beta_1\psi_1(x)\psi_2(y)$. 

Consequently, according to singular perturbation theory~\cite{khalil_nonlinear_2002}, if the above conditions hold, $\dot{V}<0$ and thus the equilibrium point of \eqref{eq3:timescale} is globally asymptotically stable. 

\end{proof}




\bibliographystyle{IEEEtran}

\bibliography{IEEEabrv,RS}

\end{document}